\documentclass{easychair}

\title{Subsumption in \flbotreg with TBoxes Is in \exptime%
	\thanks{This work was supported by the National Science Centre, Poland,
		under the OPUS 29 project no.\@ 2025/57/B/ST6/04535,
		``From Games to Algorithms: Studying Reasoning in Description Logics''.}
}

\author{
Michał Henne\inst{1}
	\and
Barbara Morawska\inst{1}
	\and
Paweł Parys\inst{2}
}

\institute{
Institute of Computer Science, University of Opole, Poland\\
	\email{$\{$michal.henne, barbara.morawska$\}$@uni.opole.pl}
	\and
Institute of Informatics, University of Warsaw, Poland.\\
	\email{parys@mimuw.edu.pl}\\
}

\authorrunning{Henne, Morawska and Parys}

\titlerunning{Subsumption in \flbotreg with TBoxes Is in \exptime}

\usepackage[utf8]{inputenc}
\usepackage{amsmath,amsthm, amssymb,paralist,aliascnt,thm-restate}
\usepackage[capitalize,nameinlink]{cleveref}
\usepackage{xspace}

\usepackage{todonotes}

\newtheorem{theorem}{Theorem}[section]
\newtheorem{lemma}[theorem]{Lemma}\Crefname{lemma}{Lemma}{Lemmata}


\theoremstyle{definition}
\newtheorem{definition}[theorem]{Definition}
\newtheorem{example}[theorem]{Example}

\theoremstyle{remark}

\Crefname{equation}{Formula}{Formulae}

\newcommand{\flo}{\ensuremath{\mathcal{FL}_0}\xspace}
\newcommand{\flbot}{\ensuremath{\mathcal{FL}_\bot}\xspace}
\newcommand{\flreg}{\ensuremath{\mathcal{FL}_{\mathit{reg}}}\xspace}
\newcommand{\flbotreg}{\ensuremath{\mathcal{FL}_{\bot\mathit{reg}}}\xspace}

\newcommand{\names}{\ensuremath{\mathsf{N}_\mathsf{C}}\xspace}
\newcommand{\roles}{\ensuremath{\mathsf{N}_\mathsf{R}}\xspace}
\newcommand{\SigmaC}{\ensuremath{\Sigma_\mathsf{C}}\xspace}
\newcommand{\SigmaR}{\ensuremath{\Sigma_\mathsf{R}}\xspace}
\newcommand{\exptime}{\mbox{\upshape{\textsc{ExpTime}}}\xspace}
\newcommand{\pspace}{\upshape{\textsc{PSpace}}\xspace}
\newcommand{\exptimebold}{E\scalebox{0.85}{XP}T\scalebox{0.85}{IME}}
\newcommand{\pspacebold}{PS\scalebox{0.85}{PACE}}
\newcommand{\ptime}{\upshape{\textsc{PTime}}}

\renewcommand{\phi}{\varphi}
\newcommand{\Nat}{\mathbb{N}}
\newcommand{\bots}{\sharp}
\newcommand{\Aa}{\mathcal{A}}
\newcommand{\Gg}{\mathcal{G}}
\newcommand{\Ii}{\mathcal{I}}
\newcommand{\Jj}{\mathcal{J}}
\newcommand{\Ll}{\mathcal{L}}
\newcommand{\Rr}{\mathcal{R}}
\newcommand{\Tt}{\mathcal{T}}
\newcommand{\pwin}{p_{\mathsf{win}}}
\newcommand{\univ}{\mathcal{U}}
\newcommand{\PA}{P_\mathsf{A}}
\newcommand{\PE}{P_\mathsf{E}}

\begin{document}
	
	\maketitle

	\begin{abstract}
	Description Logics (DLs) are a family of formal languages used for representing and reasoning about structured knowledge in terms of concepts and their relationships. The expressive power of a DL depends on the constructors available for building complex concepts.
	
In this work, we investigate subsumption in the restricted description logic \flbotreg and the related fragments \flreg, \flbot, and \flo. These formalisms support value restrictions over role names, where the subscript $\mathit{reg}$ indicates the use of regular expressions over roles.

	Subsumption between two concept descriptions in \flbotreg and \flreg is \pspace-com\-plete. When subsumption is considered with respect to a TBox (i.e., a set of axioms), the complexity increases to \exptime-complete. These results can be derived either from complexity bounds established for more expressive logics or from algorithms designed for harder reasoning problems.
	
	We reprove the \pspace-completeness result and provide a new proof of \exptime-completeness for \flreg and \flbotreg with TBoxes via a novel reduction to parity pushdown games. Our algorithm relies only on the constructs available in these logics and may therefore be implemented more easily.
	
	\end{abstract}

	\section{Introduction}
	
	Description logics constitute a family of formal languages designed to represent structured knowledge in terms of concept hierarchies and the relationships between them.
	They provide a rigorous framework for describing classes of objects, their properties, and the relations that hold among them.
	A central aim of description logics is to support precise concept definitions and to enable automated reasoning about class membership and inter-concept relationships.
	
	A fundamental reasoning task in this context is the decision problem of whether a subsumption relation holds between two given concept descriptions.
	If subsumption is decidable in a description logic, then the concept hierarchy (taxonomy) induced by the subsumption relation can be computed automatically, facilitating ontology construction and maintenance \cite{Baader2007}.

	In this paper, we present a novel algorithm for the subsumption problem in certain restricted description logics, with potential applications in ontology maintenance.
	
	Furthermore, it is often desirable to restrict the class of admissible interpretations of an ontology to those that satisfy a set of general constraints, while preserving the decidability of key reasoning tasks such as subsumption.
	These constraints can be formalized as a finite set of axioms describing relationships between concepts, known as a TBox.
	
	The subsumption problem, both with and without TBoxes, has been extensively studied across a broad spectrum of description logics.
	
	A particular description logic is characterized by the set of constructors allowed for concept definitions. Individual logics therefore differ in their expressive power, depending on the constructors provided.
	
	In this paper, we focus on description logics that admit only conjunction ($\sqcap$), the top concept ($\top$), value restrictions ($\forall r.C$), and, optionally, the bottom concept ($\bot$), interpreted as the empty set.
	These constructors are used to build complex concept descriptions from a given set of atomic concept and role names.
	
	The logic comprising the first three constructors is denoted by \flo, while the extension including $\bot$ is denoted by \flbot.
	If value restrictions are generalized to regular expressions over the alphabet of role names, the resulting logic is denoted by \flreg, and its extension with the bottom concept $\bot$ is denoted by \flbotreg.
	
	For instance, in \flo, one can describe the concept of an undergraduate student as a student who attends only undergraduate courses by the equivalence
	\[
	\texttt{UndergraduateStudent} \equiv \texttt{Student} \sqcap \forall\texttt{attends}.\texttt{UndergraduateCourse},
	\]
	where 
	\texttt{attends} is a role name relating a person and a course they attend.
	
	In \flbot, one can express the fact that professors do not attend any courses by the subsumption axiom
	\[
	\texttt{Professor} \sqsubseteq\forall\texttt{attends}.\bot.
	\]
	
	In \flreg, one can describe the more complex concept of a computer science course whose prerequisites are computer science courses,
	all prerequisites of those prerequisites are computer science courses, and so on, for any number of steps through the prerequisite relation:
	\[
	\forall\texttt{prerequisite}^*.\texttt{ComputerScienceCourse}.
	\]
	
	The computational complexity of the subsumption problem between two concept descriptions depends on the set of constructors available in the underlying logic.
	\Cref{table:subsumption} summarizes the current state of the art. The entries in bold font are those in which we are particularly interested in this paper.
	
 The first two entries in the first column indicate the existence of polynomial-time algorithms for deciding subsumption in \flo and \flbot, results that have been known for a long time (see, e.g., \cite{Baader2007}).
	
	The third and fourth entries in the first column were not stated anywhere explicitly up to our knowledge, although these complexities are consequences of other results. For example, \pspace complexity for \flreg follows from \pspace-completeness of subsumption in \flo with cyclic concept definitions interpreted in gfp-semantics \cite{Baader1990}, which is equivalent to \flreg \cite{Baader1991}. The complexity of \flbotreg follows from the characterization of 
	concept equivalence in Baader and K\"usters \cite{Baader2002}.
In this paper we use a similar characterization of subsumption as in Baader and K\"usters \cite{Baader2002} to prove the complexity of the subsumption problem  in \flreg and \flbotreg (Section~\ref{section:subsumption-no-TBox}). Another way to infer this complexity upper bound would be to refer to the result of Baader and K\"usters \cite{Baader2001a} which states that the matching problem is in \pspace for \flreg, and matching subsumes the subsumption problem. 
	
	The first entry in the second column is established by Baader, Brandt, and Lutz \cite{Baader2005} by a reduction to and from the subsumption problem in $\flo^{tf}$ with TBox \cite{Toman2005}. The description logic $\flo^{tf}$ is a variant of \flo in which all roles are interpreted as total functions. The \exptime-completeness of the problem was also established independently by Hofmann \cite{Hofmann2005}.
	
	The result mentioned in the second entry in the second column  follows from Theorem 2.2 in  Baader et al.~\cite{Baader2022}.
	
	The third and fourth entries in the second column are consequences of the fact that satisfiability in Propositional Dynamic Logic (PDL) is decidable in \exptime \cite{FischerL79}. The description logics \flreg and \flbotreg are contained in PDL
	and the TBoxes can be \emph{internalized} in a satisfiability problem as described in De Giacomo and Lenzerini \cite{Giacomo1994} (and observed in \cite{Schild1991, Baader1993}).
	
	\begin{table}
		\centering
		\renewcommand{\arraystretch}{1.3}
		\setlength{\tabcolsep}{10pt}
		\begin{tabular}{|l|c|c|}
			\hline
			\textbf{DL} & \textbf{without TBox} & \textbf{with TBox} \\
			\hline
			\flo & \ptime & \exptime-complete \\
			\hline
			\flbot & \ptime & \exptime-complete \\
			\hline
			\flreg & \textbf{\pspacebold-complete} & \textbf{\exptimebold-complete} \\
			\hline
			\flbotreg & \textbf{\pspacebold-complete} & \textbf{\exptimebold-complete} \\
			\hline
		\end{tabular}
		\caption{Complexity of the subsumption problem in various description logics.}
		\label{table:subsumption}
	\end{table}
	
	Clearly the \exptime-hardness result from the first entry in the second column ($\flo$ with TBoxes) implies \exptime-hardness of the subsumption problem in the more general logics in the last three entries.

Hence, the results presented in this paper are not new and belong to the common body of knowledge. Nevertheless, the aim of this paper is to provide a novel method for solving the subsumption problem in all the logics considered—a method that does not rely on stronger description logics or on algorithms designed to solve more difficult problems. In this way, we aim to solve the subsumption problem using only the means provided by a given logic. Since the logics considered here are very restricted, the proposed procedure may be implemented relatively easily.
	
	The approach to constructing new algorithms is to establish an equivalence between an instance of the subsumption problem with TBoxes and a suitably defined parity pushdown game associated with the instance,
	and then to apply existing techniques from game theory to solve it.
	
	Parity pushdown games arise from the combination of pushdown systems, that is, transition systems equipped with a stack,
	and parity games---two-player infinite-duration games played on finite graphs, in which the winner is determined by the parity of the priorities visited infinitely often.
	This framework has found widespread application in the verification of recursive programs and pushdown models~\cite{AlurRecursive, Walukiewicz2001},
	as well as in interprocedural static analysis frameworks based on weighted pushdown systems~\cite{Reps2005}.
	
	Solving parity pushdown games, that is, deciding the existence of a winning strategy for one of the players, is \exptime-complete~\cite{Walukiewicz2001}.
	A variant of pushdown games has been used to establish \exptime-completeness of the subsumption problem in the description logic \flo with TBoxes~\cite{Hofmann2005}.
	Motivated by this result, we present a new method for solving subsumption problems in several description logics extended with TBoxes.
	
	The remainder of this paper is structured as follows.
	\cref{section:preliminaries} recalls the necessary preliminaries.
	In \cref{section:normal-forms}, we present syntactic laws and normal forms used throughout the paper.
	\cref{section:subsumption-no-TBox} establishes \pspace-completeness of subsumption in \flreg and \flbotreg without TBoxes.
	In \cref{section:reduction}, we present a reduction from \flbotreg with TBoxes to \flreg with TBoxes, allowing us to focus on subsumption in \flreg with TBoxes.
	\cref{section:games} provides the necessary background on parity pushdown games.
	In \cref{section:reduction-to-games}, we prove our main result---\exptime-completeness of subsumption in \flreg with TBoxes---by reduction to parity pushdown games.
	Finally, \cref{section:conclusions} concludes the paper.

	\section{Preliminaries}\label{section:preliminaries}
	
	We introduce the description logics considered in this paper, starting with the most complex one, \flbotreg;
	the other logics are obtained as its syntactic fragments.

	\subsection{The logic \texorpdfstring{\flbotreg}{FL⊥reg}}
	
	Let \names and \roles be disjoint countable sets of \emph{concept names} and \emph{role names}, respectively.
	
	\emph{Regular role expressions} are just usual regular expressions built from role names using union, concatenation, and Kleene star.
	Thus, they are defined by the grammar
	\[
	E \ ::=\ \emptyset \mid \varepsilon \mid r \mid (E + E) \mid (E E) \mid E^{*}\qquad (r \in \roles).
	\]
	Then, \emph{concept descriptions} of $\flbotreg$ are generated by the following grammar:
	\[
	C \ ::=\ A \mid \top \mid \bot \mid (C \sqcap C) \mid \forall E.C \qquad (A \in \names).
	\]
	
	The semantics of \flbotreg is given by \emph{interpretations} of the form $\Ii = (\univ^\Ii, \cdot^\Ii)$, where $\univ^\Ii$ is a nonempty domain of discourse (\emph{universe}),
	and $\cdot^\Ii$ is an interpretation function mapping each concept name $A \in \names$ to a set $A^\Ii \subseteq \univ^\Ii$, and each role name $r \in \roles$ to a binary relation $r^\Ii \subseteq \univ^\Ii \times \univ^\Ii$.
	For a regular role expression, we can define its language $\Ll(E)\subseteq\roles^*$ in the usual way:
	\begin{gather*}
		\Ll(\emptyset)=\emptyset,\qquad
		\Ll(\varepsilon)=\{\varepsilon\},\qquad
		\Ll(r)=\{r\},\qquad
		\Ll(E_1 E_2) = \{w_1w_2\mid w_1\in \Ll(E_1),\ w_2\in \Ll(E_2)\},\\
		\Ll(E_1+E_2) = \Ll(E_1) \cup \Ll(E_2),\qquad
		\Ll(E^{*}) = \{ w_1w_2\dots w_n \mid n\in\Nat,\ w_1,w_2,\dots,w_n\in \Ll(E)\},
	\end{gather*}
	where $\varepsilon$ is the empty word, while $w_1w_2$ and $w_1w_2\dots w_n$ denote concatenations of words.
	Regular role expressions and complex concept descriptions are interpreted as follows, where ``$\circ$'' denotes the composition of binary relations:
	\begin{gather*}
		E^\Ii=\left\{r_1^\Ii\circ r_2^\Ii\circ\dots\circ r_n^\Ii\mid r_1r_2\dots r_n\in \Ll(E)\right\},\\
		\top^\Ii = \univ^\Ii,\qquad
		\bot^\Ii = \emptyset,\qquad
		(C \sqcap D)^\Ii = C^\Ii \cap D^\Ii,\\
		(\forall E.C)^\Ii = \left\{ x \in \univ^\Ii \mid y \in C^\Ii\mbox{ for all }y\mbox{ such that }(x,y) \in E^\Ii\right\}.
	\end{gather*}

	\subsection{Weaker logics}
	
	The logic \flbot is a syntactically limited fragment of \flbotreg, where we are not allowed to use value restrictions $\forall E.C$ for arbitrary regular expressions $E$,
	but only $\forall r.C$ for single role names $r$.
	
	The logic \flreg is obtained by syntactically limiting \flbotreg in another way: we forbid the use of~$\bot$.
	
	In order to obtain the logic \flo, one needs to apply both limitations: value restrictions are limited to $\forall r.C$, and $\bot$ cannot be used.

	\subsection{Subsumption problem}\label{section:subsumption-problem}
	
	The notions of subsumption and a TBox are defined uniformly for all description logics introduced above.
	In each case, concept descriptions are formed according to the given logic, and interpretations assign meaning to the constructors as specified.
	
	\begin{definition}
		A \emph{subsumption} between the concept descriptions $C$ and $D$ is written $C \sqsubseteq D$, and we also say that $C$ is \emph{subsumed} by $D$.
		Such a subsumption holds in an interpretation $\Ii$ if $C^\Ii\subseteq D^\Ii$.
		When no interpretation is specified, $C\sqsubseteq D$ means that $C^\Ii \subseteq D^\Ii$ holds for all interpretations $\Ii$.
		The concept descriptions $C$ and $D$ are \emph{equivalent} (written $C \equiv D$) if $C \sqsubseteq D$ and $D \sqsubseteq C$.
	\end{definition}
	
	A \emph{TBox} $\Tt$ is a finite set of axioms of the form $C \sqsubseteq D$, where $C$ and $D$ are concept descriptions.
	Such axioms express inclusion (implication) between concepts, stating that every instance of $C$ is also an instance of $D$.
	For example, $\texttt{Human} \sqsubseteq \texttt{Animal}$ states that every human is an animal.
	Sometimes one also allows axioms of the form $A \equiv B$, expressing equivalence of two concepts.
	Such axioms can be split into two subsumptions, $A \sqsubseteq B$ and $B \sqsubseteq A$, hence we assume that they are not present.
	
	Subsumption can be relativized to a TBox as follows.
	
	\begin{definition}
		Let $\Tt$ be a TBox and let $C, D$ be concept descriptions.
		\begin{compactitem}
			\item An interpretation $\Ii$ is a \emph{model} of $\Tt$ if all axioms $(E \sqsubseteq F)\in \Tt$ hold in $\Ii$.
			\item $C$ is \emph{subsumed} by $D$ with respect to $\Tt$ (written $C \sqsubseteq_\Tt D$) if $C \sqsubseteq D$ holds in every model $\Ii$ of $\Tt$.
		\end{compactitem}
	\end{definition}
	
	The \emph{subsumption problem} in the considered logic with TBoxes is defined as follows:
	\begin{compactitem}
		\item \textbf{Input:} a TBox $\Tt$ and concept descriptions $C, D$.
		\item \textbf{Output:} \texttt{Yes}, if the subsumption $C \sqsubseteq_\Tt D$ holds, and \texttt{No} otherwise.
	\end{compactitem}
	An instance of this problem is denoted $C \sqsubseteq_\Tt^? D$, with the question mark indicating that the subsumption is the \emph{goal} to be determined.
	We refer to such expressions as \emph{goal subsumptions}.
	
	If the TBox is absent (empty), we talk about the subsumption problem without TBoxes, with instances of the form $C \sqsubseteq^? D$.
	
	Notice that a \emph{set of goal subsumptions} is simply a set of independent subsumption problems, because each goal subsumption is considered independently of the others.

	\section{Syntactic laws and normal forms}\label{section:normal-forms}
	
	Before addressing the subsumption problem, in this section we present some simple syntactic laws that can be used to transform concept descriptions or to reason about subsumptions in our logics.
	
	We start by observing that the equivalence relation is a congruence: if a concept description $C$ syntactically contains a concept description $D$ (as a subterm), and $D\equiv D'$,
	then after replacing $D$ by $D'$ in the definition of $C$ we obtain a concept description $C'$ that is equivalent to $C$ (i.e., $C\equiv C'$).
	
	We now present several equivalences, which are immediate consequences of the semantics.
	First, we note that value restrictions distribute over conjunction:
	\[
	\forall E.\bigl(C_1 \sqcap C_2\bigr) \equiv \forall E.C_1 \sqcap \forall E.C_2.
	\]
	The following two equivalences, used exhaustively as rewrite rules from left to right, allow us to eliminate the $\top$ symbol (except in the case where the entire concept description is $\top$):
	\[
	\forall E.\top \equiv \top,\qquad C \sqcap \top \equiv C.
	\]
	We can also merge nested value restrictions into a single one, or introduce a value restriction when it is not present:
	\[
	\forall E_1.\forall E_2.C \equiv \forall E_1E_2.C, \qquad C\equiv\forall\varepsilon.C.
	\]
	Value restrictions involving the same concept name, as well as those involving $\bot$, may be merged:
	\[
	\forall E_1.C\sqcap\forall E_2.C \equiv \forall(E_1+E_2).C.
	\]
	
	The above transformations allow us to convert every concept description into an equivalent concept description in a normal form defined as follows.
	A concept description is in \emph{normal form} if it is of the form either $\top$, or $\forall E_1. A_1 \sqcap \cdots \sqcap \forall E_n.A_n$, or $\forall E_1. A_1 \sqcap \cdots \sqcap \forall E_n.A_n \sqcap \forall E_{n+1}.\bot$,
	where all $A_1, \dots, A_n$ are pairwise different concept names.
	A concept description of the form $\forall E.C$, where $C$ is either $\bot$ or a concept name from $\names$, is called a \emph{particle}.
	
	We also observe that for any concept name of interest, $A\in\names$ (as well as for $\bot$),
	every concept description $D$ is equivalent to a concept description in normal form containing a particle that involves $A$ (or $\bot$, respectively).
	Indeed, if such a particle is not present in $D$, we can introduce one with the empty regular expression thanks to the equivalence
	\begin{align*}
		\top\equiv\forall\emptyset.C.
	\end{align*}
	
	The next equivalence is used in \cref{section:subsumption-no-TBox}:
	\begin{align}\label{eq:bot-simplify}
		C \sqcap \bot \equiv \bot.
	\end{align}
	
	We now turn to properties of subsumption.
	Clearly subsumption is transitive: for every interpretation $\Ii$, if $C_1\sqsubseteq C_2$ and $C_2\sqsubseteq C_3$ hold in $\Ii$, then $C_1\sqsubseteq C_3$ holds as well.
	Another property is monotonicity of value restrictions: if $C\sqsubseteq D$ holds in an interpretation $\Ii$, then $\forall E.C\sqsubseteq\forall E.D$ as well, for any regular role expression $E$.
	Moreover, the following subsumption holds in every interpretation $\Ii$:
	\begin{align}\label{eq:cap-on-left}
		C\sqcap D\sqsubseteq C.
	\end{align}
	Next, observe that for every interpretation $\Ii$ we have
	\begin{align}\label{eq:cap-on-right}
		C^\Ii\subseteq(D_1\sqcap D_2)^\Ii\quad\Longleftrightarrow\quad C^\Ii\subseteq D_1^\Ii\ \land\ C^\Ii\subseteq D_2^\Ii.
	\end{align}
	Finally, for every interpretation $\Ii$ and every element $x$ of its domain we have the following equivalence (where a word $w\in\roles^*$ is treated as a regular role expression):
	\begin{align}\label{eq:in-restriction}
		x\in(\forall E.C)^\Ii \quad\Longleftrightarrow\quad \big(x\in(\forall w.C)^\Ii\mbox{ for all }w\in \Ll(E)\big).
	\end{align}

	\section{Subsumption in \texorpdfstring{\flreg}{FLreg} and \texorpdfstring{\flbotreg}{FL⊥reg} without TBoxes}\label{section:subsumption-no-TBox}
	
	In this section we solve the subsumption problem in \flreg and in \flbotreg without TBoxes.
	
	\begin{theorem}\label{theorem:no-tbox}
		The following two problems are \pspace-complete:
		\begin{compactitem}
			\item	the subsumption problem in \flreg without TBoxes, and
			\item	the subsumption problem in \flbotreg without TBoxes.
		\end{compactitem}
	\end{theorem}
	
	Consider an instance $C \sqsubseteq^? D$ of the subsumption problem.
	As already mentioned, we may assume that $C$ and $D$ are in normal form; in particular these are conjunctions of particles.
	
	We first note that a subsumption $C \sqsubseteq D_1\sqcap\dots\sqcap D_k$ holds if and only if the subsumptions $C \sqsubseteq D_i$ hold for all $i\in\{1,\dots,k\}$ (as follows from \cref{eq:cap-on-right}).
	By considering each of the latter subsumptions separately, we may assume that the right side of our goal subsumption consists of a single particle.
	
	After the above step, the considered instance of the subsumption problem in \flbotreg is of the form either $C \sqsubseteq^? \forall G.A$, where $A$ is a concept name, or $C \sqsubseteq^? \forall G.\bot$;
	the left side $C$ is a conjunction of particles involving different concept names.
	Moreover, as explained in \cref{section:normal-forms}, we may assume that $C$ contains a particle involving $\bot$, as well as a particle involving $A$ (in the case when the right side is $\forall G.A$ for $A\in\names$).
	In other words, the instance of the problem, after preparatory syntactic transformations, is of the form either
	\begin{compactitem}
		\item $\forall E.A\sqcap\forall F.\bot\sqcap C_0\sqsubseteq^?\forall G.A$, where $C_0$ does not use $\bot$ nor the concept name $A$, or
		\item $\forall F.\bot\sqcap C_0\sqsubseteq^?\forall G.\bot$, where $C_0$ does not use $\bot$.
	\end{compactitem}
	
	As a next simplification, we observe that the part $C_0$ can be removed from these subsumptions, as stated in the following lemma, whose proof can be found in  \cref{appendix:subsumption-no-TBox}.
	
	\begin{restatable}{lemma}{lemmaNoTBox}\label{lemma:no-TBox}
		For subsumptions as above we have equivalences:
		\begin{compactitem}
			\item $\forall E.A\sqcap\forall F.\bot\sqcap C_0\sqsubseteq\forall G.A$ holds if and only if $\forall E.A\sqcap\forall F.\bot\sqsubseteq\forall G.A$ holds;
			\item $\forall F.\bot\sqcap C_0\sqsubseteq\forall G.\bot$ holds if and only if $\forall F.\bot\sqsubseteq\forall G.\bot$ holds.
		\end{compactitem}
	\end{restatable}
	
	Let $\SigmaR\subseteq\roles$ contain all role names occurring in the considered concept descriptions.
	It is a finite set, hence disjunction of its elements, written simply $\SigmaR$, can be used as a regular role expression.%
	\footnote{The only reason for using $\SigmaR$ instead of $\roles$ is that $\roles$ is infinite, hence it cannot be legally used in a regular expression.}
	The next lemma (with a proof in   \cref{appendix:subsumption-no-TBox}) reduces the subsumption problem to the containment problem for languages of regular expressions.
	
	\begin{restatable}{lemma}{lemmaSubsumptionLanguage}\label{lemma:subsumption2language}
		We have the following equivalences:
		\begin{compactitem}
			\item the subsumption $\forall E.A \sqcap \forall F.\bot \sqsubseteq \forall G.A$ holds if and only if $\Ll(G)\subseteq \Ll(E+F\SigmaR^*)$;
			\item the subsumption $\forall F.\bot \sqsubseteq^? \forall G.\bot$ holds if and only if $\Ll(G) \subseteq \Ll(F\SigmaR^*)$.
		\end{compactitem}
	\end{restatable}
	
	We have thus reduced the subsumption problem in \flbotreg without TBoxes (and hence also in \flreg without TBoxes, which is a special case) to the problem of deciding inclusion between two languages given by regular expressions.
	It is known that the latter problem is \pspace-complete~\cite{Hunt1976}.
	Thus, our subsumption problems are in \pspace.
	
	In order to show \pspace-hardness, we reduce in the opposite direction.
	Given two arbitrary regular expressions $G$, $E$ for which we want to check whether language inclusion holds, we treat their letters as elements of $\roles$.
	We fix some concept name $A$, and we define $\SigmaR$ to contain all letters from $E$ and $G$.
	Taking $F=\emptyset$ we have $\Ll(E)=\Ll(E+F\SigmaR^*)$ as well as $\forall E.A\equiv\forall E.A\sqcap\forall F.\bot$.
	From \cref{lemma:subsumption2language} it thus follows that the inclusion $\Ll(G)\subseteq\Ll(E)$ (equivalently: $\Ll(G)\subseteq\Ll(E+F\SigmaR^*)$)
	is equivalent to the subsumption $\forall E.A\sqsubseteq\forall G.A$ (equivalently: $\forall E.A\sqcap\forall F.\bot\sqsubseteq\forall G.A$).
	Note that the resulting subsumption $\forall E.A\sqsubseteq\forall G.A$ does not use $\bot$, hence it belongs to both logics \flreg and \flbotreg.
	This concludes our proof of \cref{theorem:no-tbox}.
	
	We note that the correspondence between problems concerning \flreg  and language inclusion is well known and has been used, for example,
	in the context of the matching problem~\cite[Lemma~6]{Baader1998} and the unification problem~\cite[Lemma~3]{Baader2002};
	some of the arguments from the cited papers could be adapted to the present setting and reused.

	\section{Reduction from \texorpdfstring{\flbotreg}{FL⊥reg} with TBoxes to \texorpdfstring{\flreg}{FLreg} with TBoxes}\label{section:reduction}
	
	In this section, we present a reduction from the subsumption problem in \flbotreg with TBoxes to the subsumption problem in \flreg with TBoxes.
	Our construction is inspired by a reduction in Kost and Morawska~\cite{Morawska2025};
	however, their result is obtained in a different setting, namely unification, and considers a transformation from \flbot \emph{without TBoxes} to \flo with TBoxes.
	Thus, in contrast to their work, we deal with subsumption (rather than unification), include TBoxes in the source problem, and target \flreg instead of \flo.
	
	A further related approach is presented in Baader et al.~\cite{Baader2022},
	where a different reduction is used to transform a subsumption problem in \flbot with TBoxes into one in \flo with TBoxes
	(their approach requires so-called flattening of the TBox, which does not seem to work in the presence of regular expressions).
	It follows that the subsumption problem in \flbot with TBoxes belongs to \exptime.
	Our result in this paper is more general, since it concerns the more expressive logics \flbotreg and \flreg.
	
	The idea of the reduction presented below (as well as of the earlier reductions) is to simulate the constructor $\bot$ by a fresh concept name~$B$, axiomatized in a TBox.
	
	Throughout this section we consider an instance $C \sqsubseteq^?_\Tt D$ of the subsumption problem in \flbotreg with TBoxes.
	All $C$, $D$, and $\Tt$ may contain occurrences of $\bot$.
	Let $\SigmaC\subseteq\names$ and $\SigmaR\subseteq\roles$ be the sets of, respectively, all concept names and all role names occurring in $C$, $D$, or in the TBox $\Tt$.
	We fix a fresh concept name $B$, which does not occur anywhere in $C$, $D$, or $\Tt$ (i.e., $B\notin\SigmaC$).
	Then, we define:
	\begin{compactitem}
		\item $K^B$ is $K$ with all occurrences of $\bot$ replaced by $B$, for any concept description $K$;
		\item $\Tt^B$ is the TBox $\Tt$ with all occurrences of $\bot$ replaced by $B$.
	\end{compactitem}
	
	Next, we define a TBox $\Rr$, which simulates properties of $\bot$.
	It contains the following axioms:
	\begin{compactitem}
		\item $B \sqsubseteq A$, for every concept name $A \in \SigmaC$;
		\item $B \sqsubseteq \forall r.B$, for every role name $r \in \SigmaR$.
	\end{compactitem}
	
	Note that $C^B$, $D^B$, and concept descriptions in $\Tt^B \cup \Rr$ do not use the $\bot$ symbol, hence belong to the \flreg logic.
	Thus, the following theorem gives us a reduction for the subsumption problem from \flbotreg with TBoxes to \flreg with TBoxes.

	\begin{theorem}\label{theorem:fromFLBOTREGtoFLREG}
		The subsumption $C \sqsubseteq_\Tt D$ from \flbotreg holds if and only if the subsumption $C^B \sqsubseteq_{\Tt^B \cup \Rr} D^B$ from \flreg holds.
	\end{theorem}
	
	We now sketch main ideas of a proof of \cref{theorem:fromFLBOTREGtoFLREG}; missing details are contained in \cref{appendix:bot}.
	Throughout this proof, letter $\Ii$ denotes interpretations used for original concept descriptions in \flbotreg,
	and letter $\Jj$ interpretations used for concept descriptions in \flreg, obtained from those in \flbotreg by replacing $\bot$ with the concept name $B$.
	
	For the right-to-left implication, consider a model $\Ii$ of $\Tt$ in which we want to check whether $C\sqsubseteq D$ holds.
	As $\Jj$ we take a modification of $\Ii$ in which $B$ is interpreted as $\emptyset$, so in the same way as $\bot$.
	It is then easy to see that $(K^B)^\Jj=K^\Ii$ for every concept description $K$ (possibly using $\bot$, but not the concept name~$B$),
	as described by the following \lcnamecref{lemma:intFLBOTREGtointFLREG}.
	
	\begin{restatable}{lemma}{lemmaintFLBOTREGtointFLREG}\label{lemma:intFLBOTREGtointFLREG}
		Let $\Ii$ and $\Jj$ be interpretations such that
		\begin{compactitem}
			\item $\univ^\Jj = \univ^\Ii$, 
			\item $A^\Jj = A^\Ii$ for each $A \in \names\setminus\{B\}$, 
			\item $r^\Jj = r^\Ii$ for each $r \in \roles$, and 
			\item $B^\Jj = \emptyset$.
		\end{compactitem}
		Then, for every concept description $K$ in \flbotreg not using the concept name $B$ we have $(K^B)^\Jj = K^\Ii$.
	\end{restatable}

	It follows that $\Jj$ is a model of $\Tt^B\cup\Rr$, so the subsumption $C^B\sqsubseteq D^B$ holds there by assumption, which implies that $C\sqsubseteq D$ holds in $\Ii$.

	The left-to-right implication is a bit more involved.
	This time we start with a model $\Jj$ of $\Tt^B\cup\Rr$ in which we want to check whether $C^B\sqsubseteq D^B$ holds.
	Note that $\Jj$ may interpret $B$ as a nonempty set.
	First, we observe that the axioms from $\Rr$ imply the following \lcnamecref{lemma:properties-of-B}.
	
	\begin{restatable}{lemma}{lemmaPropertiesOfB}\label{lemma:properties-of-B}
		For every \flreg concept description $K$ involving only concept names from $\SigmaC\cup\{B\}$ and role names from $\SigmaR$, the subsumption $B \sqsubseteq K$ holds in every model $\Jj$ of $\Rr$.
	\end{restatable}
	
	If $B^\Jj=\univ^\Jj$, then the subsumption $C^B\sqsubseteq D^B$ holds in $\Jj$ trivially, since \cref{lemma:properties-of-B} yields $B^\Jj\subseteq(D^B)^\Jj$, and thus $(D^B)^\Jj=\univ^\Jj$.
	Otherwise, we define $\Ii$ by removing $B^\Jj$ from the universe of $\Jj$, hence also from the interpretation of every concept name and role name.
	One then has to check that $K^\Ii=(K^B)^\Jj\setminus B^\Jj$ for every concept description $K$ using only concept names from $\SigmaC$ and role names from $\SigmaR$.
	This is described by the following \lcnamecref{lemma:intFLREGtointFLBOTREG}, whose proof is by induction on the structure of $K$, and uses \cref{lemma:properties-of-B} while handling the case of value restrictions.

	\begin{restatable}{lemma}{lemmaintFLREGtointFLBOTREG}\label{lemma:intFLREGtointFLBOTREG}
	Let $\Jj$ be a model of $\Rr$ such that $B^\Jj \neq \univ^\Jj$, and let $\Ii$ be an interpretation such that
	\begin{compactitem}
		\item $\univ^\Ii = \univ^\Jj\setminus B^\Jj$,
		\item $A^\Ii = A^\Jj\setminus B^\Jj$ for each $A \in \names$, and
		\item $r^\Ii = \{(x,y)\in r^\Jj\mid x\notin B^\Jj\land y\notin B^\Jj\}$ for each $r \in \roles$.
	\end{compactitem}
	Then, for every concept description $K$ in $\flbotreg$ involving only concept names from \SigmaC and role names from \SigmaR, we have
	\[
	K^\Ii = (K^B)^\Jj\setminus B^\Jj.
	\]
	\end{restatable}

	Using \cref{lemma:intFLREGtointFLBOTREG} one can deduce that $\Ii$ is a model of $\Tt$.
	Since, by assumption, the subsumption $C\sqsubseteq_\Tt D$ holds, we obtain in particular $C^\Ii \subseteq D^\Ii$.
	Using \cref{lemma:intFLREGtointFLBOTREG} again, for $K=C$ and for $K=D$, one can then deduce that $(C^B)^\Jj\setminus B^\Jj \subseteq (D^B)^\Jj\setminus B^\Jj$.
	Taking the union with the inclusion $B^\Jj\subseteq (D^B)^\Jj$ resulting from \cref{lemma:properties-of-B}, this yields $(C^B)^\Jj \subseteq (D^B)^\Jj$.
	Hence $C^B\sqsubseteq D^B$ holds in $\Jj$.
	
	We have shown that the subsumption problem in \flbotreg with TBoxes reduces to that in \flreg with TBoxes.
	In the next part, we show how to decide the latter by reduction to the problem of solving parity pushdown games.

	\section{Parity pushdown games}\label{section:games}
	
	A parity pushdown game is an infinite-duration two-player game played on the configuration graph of a pushdown system.
	Positions of the game consist of a control state together with a stack content, and are partitioned between the two players, called Adam and Eve.
	Starting from a chosen configuration, the players move a token along transitions induced by the pushdown system, thereby generating an infinite play.
	The winner is determined by a parity condition on the sequence of control states visited during the play.
	
	Formally, a \emph{parity pushdown game} is a tuple $\Gg = (\PA, \PE, \Gamma, \Delta, \Omega)$, where
	\begin{compactitem}
		\item $\PA\uplus\PE$, which we denote $P$, is a finite set of control states, partitioned into states of Adam and Eve;
		\item $\Gamma$ is a finite stack alphabet;
		\item $\Delta \subseteq (P\times \Gamma^*) \times (P \times \Gamma^*)$ is a finite set of transition rules;
		\item $\Omega\colon P\to\Nat$ is a priority function.
	\end{compactitem}
	
	A \emph{configuration} is a pair $(p, w) \in P \times \Gamma^*$.
	The \emph{transition relation}, denoted $\to$, is a relation between configurations defined as follows: for every $(p, \alpha, q, \beta) \in \Delta$ and $w \in \Gamma^*$ there is a transition
	$$
	(p, w\alpha) \to (q, w\beta).
	$$
	Observe that the stack grows to the right (i.e., the top of the stack is the rightmost symbol of~$w$).
	Note also that transitions in $\Delta$ may remove (pop) an arbitrary number of stack symbols in a single step.
	This more general model can be effectively reduced to the standard one in which each transition removes precisely one stack symbol.
	
	A \emph{path} is a sequence of configurations, $(p_0, w_0)(p_1, w_1)(p_2, w_2)\dots$, either finite or infinite, such that each consecutive pair follows the transition relation.
	A path is called a \emph{play} if it is maximal: either infinite, or ending in a configuration $(p_n,w_n)$ from which there is no transition.
	In a finite play, the winner is given by the state of the last configuration: the player owning this state loses, as no move is available.
	In an infinite play, the winner is determined by the parity condition applied to the sequence of control states $(p_i)_{i\geq 0}$:
	Eve wins if the greatest priority occurring infinitely often in this sequence is even; otherwise, Adam wins.
	
	A (positional) \emph{strategy} for Eve is a function defined on configurations $(p, w) \in\PE \times \Gamma^*$ with at least one available move, assigning to each such configuration a valid successor.
	Strategies for Adam are defined analogously.
	A play $(p_0, w_0)(p_1, w_1)(p_2, w_2)\dots$ (finite or infinite) is consistent with a strategy $\sigma$ for Eve if, for every $i$ such that $p_i \in\PE$ and $(p_i, w_i)$ has at least one available move, we have
	$$
	(p_{i+1}, w_{i+1}) = \sigma(p_i, w_i).
	$$
	We say that Eve \emph{wins} the game $\Gg$ from a \emph{configuration $c$} if she has a strategy $\sigma$ such that every play starting from $c$ and consistent with $\sigma$ is winning for her (and analogously for Adam).
	It is well known that parity games are positionally determined, meaning that from every configuration exactly one of the players has a winning positional strategy~\cite{EmersonJutla}.

	In the next section we use the following result regarding parity pushdown games.
	
	\begin{theorem}[\cite{Walukiewicz2001}]\label{thm:parity-exptime}
		Given a parity pushdown game $\Gg$ together with an initial configuration~$c_0$,
		one can decide in \exptime whether Eve wins $\Gg$ from $c_0$.
	\end{theorem}

	\section{Reduction of subsumption problems to parity pushdown games}\label{section:reduction-to-games}
	
	In this section, we present a reduction of the subsumption problem in $\flreg$ with TBoxes to the problem of finding a winner in parity pushdown games.
	
	We shall use nondeterministic finite automata to represent regular expressions occurring in the TBox axioms.
	A \emph{nondeterministic finite automaton} (NFA) is a tuple $\Aa=(Q,\Sigma,\delta,q_0,F)$, where $Q$ is a finite set of states, $\Sigma$ is a finite input alphabet,
	$\delta\subseteq Q\times\Sigma\times Q$ is the transition relation, $q_0\in Q$ is the initial state, and $F\subseteq Q$ is the set of accepting states.
	Instead of $(p,a,q)\in\delta$, where $a\in\Sigma$, we write $p\xrightarrow{a} q$.
	A run of $\Aa$ on a word $w\in\Sigma^*$ is a sequence of transitions from $q_0$ reading the symbols of $w$.
	The automaton $\Aa$ accepts $w$ if there exists a run on $w$ ending in a state from $F$.
	The language recognized by $\Aa$, denoted $\Ll(\Aa)$, is the set of all words accepted by $\Aa$.
	
	Let $C \sqsubseteq_\Tt^? D$ be an instance of the subsumption problem in $\flreg$ with TBoxes.
	First, we introduce two fresh concept names $L$ and $R$, and define
	$$\Tt_{C,D} := \{L \sqsubseteq C,\ C\sqsubseteq L,\ R \sqsubseteq D,\ D\sqsubseteq R\},$$
	making sure that $L$ and $R$ define the same concepts as $C$ and $D$, respectively.
	It is easy to see that then\footnote{A careful reader may notice that the middle two subsumptions are not needed; it is enough to take $L \sqsubseteq C$ and $D\sqsubseteq R$.}
	$$C \sqsubseteq_\Tt D\qquad\text{if and only if}\qquad L \sqsubseteq_{\Tt \cup \Tt_{C,D}} R.$$
	Hence, from now on, it suffices to consider a subsumption problem of the form $L \sqsubseteq_\Tt R$ with $L,R$ being concept names, where, by abuse of notation,
	$\Tt$ already includes the additional axioms from $\Tt_{C,D}$.
	
	Moreover, we assume that every concept occurring in the TBox $\Tt$ is in normal form, that is, it is a conjunction of particles of the form $\forall E.A$,
	where $E$ is a regular expression over the alphabet $\roles$, and $A\in\names$ is a concept name.
	
	Going further, we transform every axiom of the form $\forall G_1.A_1 \sqcap \ldots \sqcap \forall G_n.A_n \sqsubseteq \forall H_1.B_1 \sqcap \ldots \sqcap \forall H_k.B_k$ from the TBox $\Tt$
	into an equivalent set of axioms with a single particle on the right side:
	$$\forall G_1.A_1 \sqcap \ldots \sqcap \forall G_n.A_n \sqsubseteq \forall H_j.B_j \qquad \text{for } j\in\{ 1, \ldots, k\}.$$
	
	Let \SigmaC and \SigmaR be the finite sets of, respectively, all concept names and all role names occurring in $\Tt$;
	in particular \SigmaC contains the two special concept names $L$ and $R$.
	Based on the TBox $\Tt$ we construct a pushdown game $\Gg_\Tt=(\PA,\PE,\Gamma,\Delta,\Omega)$.
	As the stack alphabet, we take $\Gamma=\SigmaR\cup\{\bots\}$, where $\bots$ is a special symbol, which will mark the bottom of the stack.
	For each concept name $A\in \names$, we introduce a state owned by Eve, $p_A\in\PE$.
	We designate $c_0=(p_R,\bots)$ as the initial configuration of the game.
	
	For every TBox axiom $\tau$, we proceed as follows.
	The axiom $\tau$ is of the form $\forall G_1.A_1 \sqcap \cdots \sqcap \forall G_n.A_n\sqsubseteq\forall H.B$.
	We construct a nondeterministic finite automaton $\Aa_H=(Q_H,\allowbreak \SigmaR,\allowbreak \delta_H,\allowbreak q_{H,0},\allowbreak F_H)$ such that $\Ll(\Aa_H)=\Ll(H)$.
	For each state $q\in Q_H$, we introduce a game state owned by Eve, $q\in\PE$.
	We then add transitions that allow to enter the automaton in any accepting state:
	$$(p_B,\varepsilon)\to(q,\varepsilon)\qquad\mbox{if }q\in F_H.$$
	For every transition $p\xrightarrow{r} q$ of the automaton $\Aa_H$, we add the reverse of this transition to the game:
	$$(q,r)\to (p,\varepsilon).$$
	We also introduce a new state $p_\tau\in\PA$, and add a transition that allows to reach it from the initial state:
	$$(q_{H,0},\varepsilon)\to (p_\tau,\varepsilon).$$
	This way, Eve can go from $p_B$ to $p_\tau$ while removing from the stack any word that matches $H$ (the word is read left-to-right, that is, bottom-to-top).
	
	Next, for every particle $\forall G_i.A_i$ on the left side of $\tau$, we allow Adam to go from $p_\tau$ to $p_{A_i}$ while inserting on the stack any word that matches $G_i$.
	To this end, for every $i\in\{1,\dots,n\}$ we construct a nondeterministic finite automaton $\Aa_{G_i}=(Q_{G_i},\SigmaR,\delta_{G_i},q_{G_i,0},F_{G_i})$ such that $\Ll(\Aa_{G_i})=\Ll(G_i)$,
	and we continue as follows.
	First, for each state $q\in Q_{G_i}$, we introduce a game state owned by Adam, $q\in\PA$.
	Then, we add a transition that allows Adam to enter the initial state of $\Aa_{G_i}$:
	$$(p_\tau,\varepsilon)\to (q_{G_i,0},\varepsilon).$$
	Next, for every transition $p\xrightarrow{r} q$ of $\Aa_{G_i}$, we add an analogous transition to the game:
	$$(p,\varepsilon)\to (q,r).$$
	Finally, for each accepting state $q\in F_{G_j}$ we add the transition
	$$(q,\varepsilon)\to (p_{A_j},\varepsilon).$$
	
	Of course we assume that state sets of all the constructed automata are disjoint, so that the automata do not interfere with each other, after adding them to the game.
	
	Success is reached when the current obligation becomes $L$ and no role symbol remains on the stack, that is, at the configuration $(p_L,\bots)$.
	To realize this, we introduce an artificial winning state $\pwin\in\PA$, and we add the transition
	$$(p_L,\bots)\to(\pwin,\bots).$$
	It is a state of Adam without any successors, so reaching it causes that Eve wins.
	
	The priority function $\Omega$ is defined as follows:
	$$\Omega(p)=
	\begin{cases}
		0,& \text{if } p\in\PA,\\
		1,& \text{if } p\in\PE.
	\end{cases}$$
	Such a definition causes that if Adam stays in his states (i.e., in some automaton $\Aa_{G_i}$) forever, then Eve wins;
	otherwise, Eve has to reach $\pwin$ in a finite time, as otherwise infinitely many states of priority $1$ would be produced.
	
	\begin{theorem}\label{theorem:main-reduction}
		Let $L \sqsubseteq_\Tt^? R$ be an instance of the subsumption problem in $\flreg$ with TBoxes (after the preliminary transformation described above),
		and let $\Gg_\Tt$ be the parity pushdown game associated with this instance by the above construction.
		Then Eve has a winning strategy from the configuration $c_0=(p_R,\bots)$ if and only if $L \sqsubseteq_\Tt R$.
	\end{theorem}
	
	We prove separately two implications.
	
	For the left-to-right implication assume that $\sigma$ is a winning strategy for Eve from $(p_R,\bots)$, and suppose towards a contradiction that $L \not\sqsubseteq_\Tt R$.
	Then there exists a model $\Ii$ of $\Tt$ and an element $x\in L^\Ii$ such that $x\notin R^\Ii$.
	Fix one such element.
	
	We construct a play starting from $(p_R,\bots)$ and consistent with $\sigma$, but won by Adam.
	This yields a contradiction.
	We construct the play so that the following invariant is preserved:
	whenever the current configuration is of the form $(p_B,\bots u)$, we have $x\notin (\forall u.B)^\Ii$.
	
	Recall that $x\notin R^\Ii=(\forall\varepsilon.R)^\Ii$, meaning that the initial configuration $(p_R,\bots)$ satisfies the invariant.
	
	Assume therefore that the game is currently in a configuration $(p_B,\bots u)$ satisfying the invariant.
	First, if no move is available to Eve from that configuration, then she loses and we are done.
	Second, observe that $(p_B,\bots u)\neq(p_L,\bots)$.
	Indeed, for $B=L$ and $u=\varepsilon$ the invariant would say that $x\notin(\forall\varepsilon.L)^\Ii=L^\Ii$, contradicting the assumption that $x\in L^\Ii$.
	Thus, the transition to the state $\pwin$ is not available.
	
	It follows that from the configuration $(p_B,\bots u)$, Eve chooses some axiom
	$$
	\tau=\big(\forall G_1.A_1 \sqcap \cdots \sqcap \forall G_n.A_n \sqsubseteq \forall H.B\big)\in\Tt
	$$
	and moves to the part of the game corresponding to this axiom.
	Every transition in this part removes a symbol from the stack, so the play cannot stay in this part for an infinite time.
	Moreover, all states in this part belong to Eve, so if the play gets stuck there, then Adam wins, as we wanted.
	
	The remaining case is that the play reaches the state $p_\tau$ after some time.
	By construction of the game, this could happen only if some word $v\in\Ll(\Aa_H)=\Ll(H)$ was removed from the stack.
	Consequently, the original word $u$ was necessarily of the form $u=u'v$ for some $u'\in\SigmaR^*$, and after this phase of the play the configuration reached is $(p_\tau,\bots u')$.
	Recall that $\tau$, as an axiom from $\Tt$, holds in $\Ii$.
	Since $v\in\Ll(H)$, all elements of $(\forall H.B)^\Ii$ are elements of $(\forall v.B)^\Ii$ (as follows from \cref{eq:in-restriction}),
	so the subsumption $\forall G_1.A_1 \sqcap \cdots \sqcap \forall G_n.A_n \sqsubseteq \forall v.B$ holds in $\Ii$ as well.
	By monotonicity of value restriction, also the subsumption $\forall u'.(\forall G_1.A_1 \sqcap \cdots \sqcap \forall G_n.A_n)\sqsubseteq\forall u'.\forall v.B$, 
	which can be syntactically rewritten as $\forall u'G_1.A_1 \sqcap \cdots \sqcap \forall u'G_n.A_n\sqsubseteq\forall u.B$, holds in $\Ii$.
	Recalling the invariant $x\notin (\forall u.B)^\Ii$, this subsumption implies that $x\notin(\forall u'G_1.A_1 \sqcap \cdots \sqcap \forall u'G_n.A_n)^\Ii$,
	hence $x\notin(\forall u'G_j.A_j)^\Ii$ for some $j\in\{1,\dots,n\}$.
	Using \cref{eq:in-restriction} we observe that then $x\notin(\forall u'w.A_j)^\Ii$ for some word $w\in\Ll(G_j)$ (where we note that all words in $\Ll(u'G_j)$ are of the form $u'w$ for some $w\in\Ll(G_j)$).
	
	Adam now chooses the transition from $p_\tau$ leading to the initial state $q_{G_j,0}$ of the automaton $\Aa_{G_j}$, and then follows an accepting run of $\Aa_{G_j}$ on the word $w$.
	Since all states in this part of the game belong to Adam, he can enforce this finite fragment of the play.
	By construction of the game, the reached configuration is $(p_{A_j},\bots u'w)$, and the invariant holds again: $x\notin (\forall u'w.A_j)^\Ii$.
	From this point on, we continue the construction of the play in the same way, preserving the invariant.
	
	Note that if we continue this construction forever, we obtain a play that visits infinitely many configurations of the form $(p_B,\bots u)$, having priority $1$;
	such a play is winning for Adam.
	Concluding, assuming $L\not\sqsubseteq_\Tt R$ we were able to construct a play won by Adam and consistent with the winning strategy of Eve, leading to a contradiction.
	Therefore $L \sqsubseteq_\Tt R$.
	
	We now come to the right-to-left implication from the theorem: assuming $L \sqsubseteq_\Tt R$, we show that Eve wins from the configuration $(p_R,\bots)$.
	
	We define an interpretation $\Ii$ as follows: the universe is $\univ^\Ii=\SigmaR^*$, for each role $r\in\SigmaR$ we take $r^\Ii=\{(u,ur)\mid u\in\SigmaR^*\}$,
	and for each concept name $A\in\SigmaC$ we take
	\[
	A^\Ii=\{u\in\SigmaR^* \mid \text{Eve wins from }(p_A,\bots u)\}.
	`	\]
	Interpretation of role names from $\roles\setminus\SigmaR$ and concept names from $\names\setminus\SigmaC$ is irrelevant.
	
	We show that $\Ii$ is a model of $\Tt$.
	Take an arbitrary axiom
	\[\tau=(\forall G_1.A_1 \sqcap \cdots \sqcap \forall G_n.A_n \sqsubseteq \forall H.B)\in\Tt,\]
	and an arbitrary word $u\in\SigmaR^*$ such that $u\in (\forall G_1.A_1 \sqcap \cdots \sqcap \forall G_n.A_n)^\Ii$.
	This means that for every $j\in\{1,\dots,n\}$ and every $w\in \Ll(G_j)$ we have $uw\in A_j^\Ii$, that is, Eve wins from configuration $(p_{A_j},\bots uw)$.
	We must show that $u\in(\forall H.B)^\Ii$, that is, that for every $v\in \Ll(H)$, Eve wins from the configuration $(p_B,\bots uv)$.
	
	Fix an arbitrary $v\in \Ll(H)$, and consider the configuration $(p_B,\bots uv)$.
	Since parity games are positionally determined, it is enough to show that Adam does not win from $(p_B,\bots uv)$.
	Suppose towards a contradiction that Adam wins from $(p_B,\bots uv)$, and fix some his winning strategy $\sigma_\mathsf{A}$.
	We construct a play consistent with $\sigma_\mathsf{A}$ and won by Eve, contradicting the fact that $\sigma_\mathsf{A}$ is winning for Adam.
	
	First, Eve chooses the axiom $\tau$, enters the part of the game corresponding to $\Aa_H$, and follows backwards an accepting run of $\Aa_H$ on the word $v$.
	By construction, this removes exactly the suffix $v$, and the play reaches $(p_\tau,\bots u)$.
	At this point Adam chooses some index $j\in\{1,\dots,n\}$, and the play enters the part corresponding to $\Aa_{G_j}$.
	There are now three possibilities.	
	
	If Adam stays forever in states belonging to $\PA$, then the resulting infinite play is winning for Eve, since only priority $0$ occurs infinitely often.
	
	If Adam gets stuck in a state belonging to $\PA$, then the resulting finite play is winning for Eve, because a finite play ending in a position of Adam is losing for Adam.
	
	The remaining possibility is that Adam follows an accepting run of $\Aa_{G_j}$ on some word $w\in\Ll(\Aa_{G_j})=\Ll(G_j)$.
	In this case the play reaches the configuration $(p_{A_j},\bots uw)$.
	As explained above, our assumption on $u$ implies that $uw\in A_j^\Ii$, meaning that Eve wins from $(p_{A_j},\bots uw)$.
	Let $\sigma_\mathsf{E}$ be her winning strategy from this configuration.
	We define the remaining part of the play following both $\sigma_\mathsf{A}$ and $\sigma_\mathsf{E}$.
	Because $\sigma_\mathsf{E}$ is winning for Eve, the suffix of the play starting from $(p_{A_j},\bots uw)$ is won by Eve, hence the whole play as well.
	
	We have thus shown that Adam cannot win from $(p_B,\bots uv)$; Eve wins there.
	Since $v\in \Ll(H)$ was arbitrary, we conclude that $u\in (\forall H.B)^\Ii$.
	Thus $\tau$ holds in $\Ii$.
	Since $\tau$ was an arbitrary axiom of $\Tt$, we obtain that $\Ii$ is a model of $\Tt$.
	
	Finally, Eve trivially wins from $(p_L,\bots)$, because there is a transition $(p_L,\bots)\to (\pwin,\bots)$, and $\pwin$ is a dead-end state owned by Adam.
	Hence $\varepsilon\in L^\Ii$.
	By the assumed subsumption $L \sqsubseteq_\Tt R$, this implies that $\varepsilon\in R^\Ii$.
	Therefore, by the definition of $R^\Ii$, Eve wins from $(p_R,\bots)$.
	
	\begin{theorem}
		The subsumption problem in each of the logics \flo, \flbot, \flreg, and \flbotreg, with TBoxes, is \exptime-complete.
	\end{theorem}
	
	\begin{proof}
		For the lower bound, recall that the subsumption problem in $\flo$ with TBoxes is \exptime-hard~\cite{Hofmann2005}.
		Since $\flo$ is a fragment of each of the logics $\flbot$, $\flreg$, and $\flbotreg$, it follows immediately that the subsumption problem with TBoxes is \exptime-hard for all of them.
		
		It remains to prove membership in \exptime.
		
		For $\flreg$, this follows from the previous theorem: every instance of the subsumption problem in $\flreg$ with TBoxes can be reduced in polynomial time to a parity pushdown game,
		and parity pushdown games are solvable in \exptime, as stated in \cref{thm:parity-exptime}.
		
		For $\flbotreg$, we use \cref{theorem:fromFLBOTREGtoFLREG}, which yields a polynomial reduction from an instance of the subsumption problem in $\flbotreg$ with TBoxes
		to an equivalent instance of the subsumption problem in $\flreg$ with TBoxes.
		Hence the subsumption problem in $\flbotreg$ with TBoxes also belongs to \exptime.
		
		Finally, \flo and \flbot are fragments of $\flbotreg$.
		Therefore subsumption in \flo with TBoxes and in \flbot with TBoxes are special cases of subsumption in $\flbotreg$ with TBoxes, and thus also belongs to \exptime
		(this also follows from earlier work~\cite{Baader2005,Hofmann2005,Baader2022}).
		
		Consequently, in each of the logics \flo, \flbot, \flreg, and \flbotreg, the subsumption problem with TBoxes is \exptime-complete.
	\end{proof}
	
	\begin{example}
		We illustrate the construction of the parity pushdown game associated with the subsumption problem
		$$
		\forall rrr^*.A \sqcap B \sqsubseteq \forall rr^*.A
		$$
		with the TBox containing the axiom $B\sqsubseteq\forall r.A$.
		Since this instance already belongs to $\flreg$, no preliminary reduction from $\flbotreg$ is needed.
		
		Following the general construction, we introduce two fresh concept names $L$ and $R$, and replace the original problem by the equivalent one $L \sqsubseteq_\Tt R$,
		where $\Tt$ consists of the following axioms, displayed here already after switching to normal forms, and splitting subsumptions with multiple particles on the right side:
		\begin{gather*}
			\forall\varepsilon.B\sqsubseteq\forall r.A,\qquad
			\forall\varepsilon.L \sqsubseteq \forall rrr^*.A,\qquad
			\forall\varepsilon.L \sqsubseteq \forall\varepsilon.B,\\
			\forall rrr^*.A \sqcap \forall\varepsilon.B \sqsubseteq \forall\varepsilon.L,\qquad
			\forall\varepsilon.R \sqsubseteq \forall rr^*.A,\qquad
			\forall rr^*.A \sqsubseteq \forall\varepsilon.R.
		\end{gather*}
		
		We now describe the corresponding parity pushdown game.
		The concept-name states are $p_L,p_R,p_A,p_B$, all owned by Eve, and the initial configuration is $(p_R,\bots)$.
		
		The nondeterministic finite automata used in the construction are the following:
		\begin{compactitem}
			\item an automaton $\Aa_{rr^*}$ recognizing $\Ll(rr^*)=\{r^n\mid n\geq 1\}$,
			\item an automaton $\Aa_r$ recognizing $\Ll(r)=\{r\}$,
			\item an automaton $\Aa_{rrr^*}$ recognizing $\Ll(rrr^*)=\{r^n\mid n\geq 2\}$,
			\item and the trivial automaton for $\varepsilon$, used for particles of the form $\forall \varepsilon.L$.
		\end{compactitem}
		
		The game starts in the configuration $(p_R,\bots)$.
		Eve first chooses the axiom
		$$
		\tau_1=(\forall rr^*.A \sqsubseteq R).
		$$
		Since the right side is just $\forall\varepsilon.R$, this corresponds to the regular expression $\varepsilon$, so no symbol is removed from the stack.
		The play reaches the state $p_{\tau_1}$.
		From $p_{\tau_1}$, Adam must choose the unique particle on the left side, namely $\forall rr^*.A$.
		To do so, he enters the part of the game corresponding to $\Aa_{rr^*}$ and generates some word from the language $\Ll(rr^*)$.
		Hence, he chooses an arbitrary word $r^n$, where $n\geq 1$, and the play reaches the configuration $(p_A,\bots r^n)$.
		Now Eve must reduce the current obligation to $L$.
		At this point she distinguishes two cases.
		
		\emph{Case 1: $n=1$.}
		Then the current configuration is $(p_A,\bots r)$.
		Eve chooses the axiom
		$$
		\tau_2=(\forall\varepsilon.B \sqsubseteq \forall r.A).
		$$
		Using the automaton $\Aa_r$, she removes the symbol $r$ from the top of the stack and reaches the state $p_{\tau_2}$ with stack $\bots$.
		The only particle on the left side is $\forall \varepsilon.B$, so Adam inserts the empty word and the play reaches $(p_B,\bots)$.
		Here, Eve chooses the axiom $\tau_3=(\forall\varepsilon.L \sqsubseteq \forall\varepsilon.B)$, reaches the state $p_{\tau_3}$ while removing the empty word from the stack,
		from which Adam can only inserts the empty word and go to $(p_L,\bots)$.
		
		\emph{Case 2: $n\geq 2$.}
		Then the current configuration is $(p_A,\bots r^n)$.
		Eve chooses the axiom
		$$
		\tau_4=(\forall \varepsilon.L \sqsubseteq \forall rrr^*.A).
		$$
		Since $r^n\in \Ll(rrr^*)$ whenever $n\geq 2$, she can use the automaton $\Aa_{rrr^*}$ to remove the whole suffix $r^n$ from the stack and reach a state $p_{\tau_4}$ with stack $\bots$.
		As above, Adam is then forced to move to $(p_L,\bots)$.
		
		In both cases Eve reaches the success configuration $(p_L,\bots)$, from which she moves to the designated winning state.
		Therefore Eve has a winning strategy from the initial configuration $(p_R,\bots)$.
		
		This example reflects exactly why the original subsumption is valid.
		The formula $\forall rr^*.A$ requires that all $r$-paths lead to elements satisfying $A$.
		Adam may therefore choose any word $r^n$ with $n\geq 1$.
		If $n=1$, Eve answers by the conjunct $\forall r.A$; if $n\geq 2$, she answers by the conjunct $\forall rrr^*.A$.
		Thus the two conjuncts on the left side cover all words in $\Ll(rr^*)$, and the subsumption holds.
	\end{example}

	\section{Conclusions}\label{section:conclusions}
	
	We have shown that the subsumption problem in \flreg extended with TBoxes is \exptime-complete.
	This result is obtained via a reduction from subsumption to parity pushdown games, which are known to be solvable in exponential time.
	
	Based on this reduction and existing algorithms for solving such games, a decision procedure for subsumption in \flreg with TBoxes can be derived in a straightforward manner.
	An implementation of this procedure is planned.
	In particular, we intend to employ existing solvers for parity pushdown games, such as PDSolver~\cite{Hague2010}.
	
	The approach also applies to \flbotreg due to the reduction presented in \cref{section:reduction}.
	Moreover, since \flo is a sublogic of \flreg, the procedure is also applicable to subsumption in \flo with TBoxes.
	
	An implementation, \emph{FL$_0$wer}, for subsumption in \flo with TBoxes already exists and is described in Baader et al.~\cite{Baader2022}.
	\emph{FL$_0$wer} is based on a different algorithmic approach.
	It is therefore of interest to investigate experimentally whether the proposed implementation performs better or worse than \emph{FL$_0$wer} on subsumption problems in \flo with TBoxes.
	Nevertheless, the presented method provides a general framework for deciding the corresponding reasoning problems in the other logics considered in this work.
	
	A further direction for future work is the extension of the presented techniques to other reasoning tasks in \flo, \flbot, \flreg, and \flbotreg with TBoxes, such as concept matching and unification.

	\bibliographystyle{plain}
	\bibliography{LPARsubmission}
	\newpage
	\appendix
	\section{Missing proofs for subsumption in \texorpdfstring{\flreg}{FLreg} and \texorpdfstring{\flbotreg}{FL⊥reg} without TBoxes}\label{appendix:subsumption-no-TBox}
	
	Recall that \cref{lemma:no-TBox} talks about subsumptions of the form either $\forall E.A\sqcap\forall F.\bot\sqcap C_0\sqsubseteq\forall G.A$ or $\forall F.\bot\sqcap C_0\sqsubseteq\forall G.\bot$,
	where the concept description $C_0$ does not use $\bot$ and, in case of the first form, $C_0$ also does not use $A$.
	
	\lemmaNoTBox*
	
	\begin{proof}
		The right-to-left implications are immediate consequence of \cref{eq:cap-on-left} and transitivity of subsumption.
		
		For the left-to-right implication in the first item, consider any interpretation $\Ii$ in which we want to check whether $\forall E.A\sqcap\forall F.\bot\sqsubseteq\forall G.A$ holds.
		Let $\Jj$ be a modification of $\Ii$ in which we map all concept names other than $A$ to the whole domain $\univ^\Ii$.
		Such a change does not alter the semantics of our concept descriptions not using those concept names, that is,
		$(\forall E.A\sqcap\forall F.\bot)^\Ii=(\forall E.A\sqcap\forall F.\bot)^\Jj$ and $(\forall G.A)^\Jj=(\forall G.A)^\Ii$.
		On the other hand, for each concept name $B$ other than $A$ and for each regular role expression $E$ the assumption $B^\Jj=\univ^\Jj$ implies that $(\forall E.B)^\Jj=\univ^\Jj$.
		We have $C_0^\Jj=\univ^\Jj$, hence $(\forall E.A\sqcap\forall F.\bot)^\Jj=(\forall E.A\sqcap\forall F.\bot\sqcap C_0)^\Jj$.
		The inclusion $(\forall E.A\sqcap\forall F.\bot\sqcap C_0)^\Jj \subseteq \forall (G.A)^\Jj$, which holds by assumption, implies the thesis.
		
		The left-to-right implication in the second item can be shown analogously.
		The only difference is that now there is no special concept name $A$; in $\Jj$ we map all concept names to the whole domain.
		Then $C_0$, not using $\bot$, again becomes interpreted as the whole domain, and we can conclude as above.
	\end{proof}
	
	\lemmaSubsumptionLanguage*
	
	\begin{proof}
		We start with the right-to-left implication from the first item.
		Assume that $\Ll(G)\subseteq \Ll(E+F\SigmaR^*)$.
		As a direct consequence of the semantics (or using the equivalence from \cref{eq:in-restriction}) we obtain the subsumption
		\[
		\forall (E+F\SigmaR^*).A\sqsubseteq \forall G.A.
		\]
		Due to \cref{eq:cap-on-left}, the subsumption holds even more if we extend the left side by an additional particle:
		\[
		\forall (E+F\SigmaR^*).A\sqcap\forall F.\bot\sqsubseteq \forall G.A.
		\]
		Using distributive laws we can change this subsumption into
		\[
		\forall E.A\sqcap\forall F.(\forall \SigmaR^*.A\sqcap\bot)\sqsubseteq \forall G.A.
		\]
		This gives the desired subsumption after removing $\forall \SigmaR^*.A$, which can be done thanks to \cref{eq:bot-simplify}.
		
		The right-to-left implication from the second item is shown similarly.
		Namely, as a consequence of $\Ll(G) \subseteq \Ll(F\SigmaR^*)=\Ll(F\SigmaR^*+F)$ we obtain the subsumption
		\[
		\forall (F\SigmaR^*+F).\bot\sqsubseteq \forall G.\bot,
		\]
		which can be rewritten as
		\[
		\forall F.(\forall \SigmaR^*.\bot\sqcap\bot)\sqsubseteq \forall G.\bot.
		\]
		This gives the desired subsumption after removing $\forall \SigmaR^*.\bot$, which can be done thanks to \cref{eq:bot-simplify}.
		
		In order to prove the opposite implications, we consider an interpretation $\Ii$ with domain $\univ^\Ii=\roles^*\setminus \Ll(F\SigmaR^*)$,
		where $r^\Ii=\{(w,wr)\mid w,wr\in\univ^\Ii\}$ for all $r\in\roles$ and, in case of the first item, $A^\Ii=\Ll(E)\cap\univ^\Ii$;
		interpretation of other concept names is irrelevant.
		Note that if a word $w\in\SigmaR^*$ belongs to $\univ^\Ii$, then all its prefixes as well.
		Thus, for any regular expression $H$ using only role names from $\SigmaR$ (in particular for $E$, $F$, and $G$) and any $w\in\roles^*$
		we have $(\varepsilon,w)\in H^\Ii$ if and only if $w\in \Ll(H)\cap\univ^\Ii$.
		This implies, for each concept description $C$, that
		\[
		\varepsilon\in(\forall H.C)^\Ii\quad\Longleftrightarrow\quad \Ll(H)\cap\univ^\Ii\subseteq C^\Ii.\tag{$\diamondsuit$}
		\]
		
		Concentrate now on the first item from the lemma statement.
		Because $\Ll(E)\cap\univ^\Ii=A^\Ii\subseteq A^\Ii$ and $\Ll(F)\cap\univ^\Ii=\emptyset\subseteq\bot^\Ii$, Equivalence ($\diamondsuit$)
		gives us $\varepsilon\in(\forall E.A)^\Ii\cap(\forall F.\bot)^\Ii=(\forall E.A\sqcap\forall F.\bot)^\Ii\subseteq(\forall G.A)^\Ii$, with the last inclusion by the assumed subsumption.
		Using Equivalence ($\diamondsuit$) again, this implies $\Ll(G)\cap\univ^\Ii\subseteq A^\Ii=\Ll(E)\cap\univ^\Ii$.
		Clearly both $\Ll(G)$ and $\Ll(E)$ are included in $\roles^*=\univ^\Ii\cup\Ll(F\SigmaR^*)$,
		so by adding $\Ll(F\SigmaR^*)$ to both sides of the above inclusion we obtain $\Ll(G)\subseteq \Ll(G)\cup\Ll(F\SigmaR^*)\subseteq\Ll(E)\cup\Ll(F\SigmaR^*)=\Ll(E+F\SigmaR^*)$, as claimed.
		
		In the second item, Equivalence ($\diamondsuit$) applied to the inclusion $\Ll(F)\cap\univ^\Ii\subseteq\bot^\Ii$ gives us $\varepsilon\in(\forall F.\bot)^\Ii$,
		which by the assumed subsumption implies $\varepsilon\in(\forall G.\bot)^\Ii$.
		Using Equivalence ($\diamondsuit$) once again, this gives us $\Ll(G)\cap\univ^\Ii\subseteq\bot^\Ii=\emptyset$.
		Recalling the definition of $\univ^\Ii$, this can be restated as $\Ll(G)\subseteq \Ll(F\SigmaR^*)$, which is the claimed inclusion.
	\end{proof}

	\section{Proof of \texorpdfstring{\cref{theorem:fromFLBOTREGtoFLREG}}{Theorem \ref{theorem:fromFLBOTREGtoFLREG}}}\label{appendix:bot}
	
	We present here a full proof of \cref{theorem:fromFLBOTREGtoFLREG}.
	
	First, we prove \cref{lemma:intFLBOTREGtointFLREG}, which is used in a setting with an arbitrary interpretation $\Ii$ (the one for \flbotreg)
	and with interpretation $\Jj$ defined based on $\Ii$.

	\lemmaintFLBOTREGtointFLREG*
	
	\begin{proof}
		We proceed by induction on the structure of $K$.
		
		If $K = A$ is a concept name, then $K^B = A$ and $A\neq B$.
		By assumption we have $A^\Jj = A^\Ii$, hence $(K^B)^\Jj = A^\Jj = A^\Ii = K^\Ii$.
		
		If $K = \top$, then $K^B = \top$.
		Since $\univ^\Jj = \univ^\Ii$, we obtain $(K^B)^\Jj = \top^\Jj = \univ^\Jj = \univ^\Ii = \top^\Ii = K^\Ii$.
		
		If $K = \bot$, then $K^B = B$.
		By assumption we have $B^\Jj = \emptyset$, and therefore $(K^B)^\Jj = B^\Jj = \emptyset = \bot^\Ii = K^\Ii$.
		
		If $K = K_1 \sqcap K_2$, then $K^B = K_1^B \sqcap K_2^B$. By the induction hypothesis we have $(K_1^B)^\Jj = K_1^\Ii$ and $(K_2^B)^\Jj = K_2^\Ii$.
		Hence $(K^B)^\Jj = (K_1^B \sqcap K_2^B)^\Jj = (K_1^B)^\Jj \cap (K_2^B)^\Jj = K_1^\Ii \cap K_2^\Ii = (K_1 \sqcap K_2)^\Ii = K^\Ii$.
		
		Finally, if $K = \forall E.M$ for a regular role expression $E$, then $K^B = \forall E.M^B$. 	
		By semantics of the value restriction we have
		\[
		(K^B)^\Jj=(\forall E.M^B)^\Jj = \big\{ x \in \univ^\Jj \mid y\in (M^B)^\Jj\mbox{ for all }y\mbox{ such that }(x,y) \in E^\Jj\big\}.
		\]
		We have $\univ^\Jj = \univ^\Ii$ by assumption, $(M^B)^\Jj=M^\Ii$ by the induction hypothesis, and $E^\Jj = E^\Ii$ because $r^\Jj=r^\Ii$ for every role name $r$.
		It follows that 
		\[
		(K^B)^\Jj = \big\{x \in \univ^\Ii \mid y\in (M^B)^\Ii\mbox{ for all }y\mbox{ such that }(x,y) \in E^\Ii\big\} = (\forall E.M)^\Ii = K^\Ii.\tag*{\qedhere}
		\]
	\end{proof}
	
	\cref{lemma:intFLBOTREGtointFLREG} allows us to obtain the the right-to-left implication of \cref{theorem:fromFLBOTREGtoFLREG}.
	In this implication we assume that the subsumption $C^B \sqsubseteq_{\Tt^B \cup\Rr} D^B$ holds, and we have to show that then $C \sqsubseteq_\Tt D$ also holds.
	To this end, take an arbitrary model $\Ii$ of the TBox $\Tt$, and consider a corresponding interpretation~$\Jj$ as specified in \cref{lemma:intFLBOTREGtointFLREG}.
	
	To see that $\Jj$ is a model of $\Tt^B$, take any axiom from this TBox.
	It is of the form $K^B\sqsubseteq M^B$ for some axiom $(K\sqsubseteq M)\in\Tt$.
	The latter axiom holds in $\Ii$, that is $K^\Ii\subseteq M^\Ii$, and from \cref{lemma:intFLBOTREGtointFLREG} we know that $(K^B)^\Jj = K^\Ii$ and $(M^B)^\Jj = M^\Ii$
	(where it is important that $B$, selected as a fresh concept, does not occur in $K$ nor in $M$).
	It follows immediately that $(K^B)^\Jj \subseteq (M^B)^\Jj$.
	
	Next, observe that for every concept name $A \in \SigmaC$ we have $B^\Jj \subseteq A^\Jj$, and for every role name $r \in \SigmaR$ we have $B^\Jj \subseteq (\forall r.B)^\Jj$,
	since $B^\Jj=\emptyset$ and the empty set is a subset of any set.
	It follows that $\Jj$ satisfies all axioms from $\Rr$.
	Thus, $\Jj$ is a model of $\Tt^B \cup\Rr$.
	
	By the assumptions of the theorem, $C$ is subsumed by $D$ in every model of $\Tt^B \cup\Rr$, in particular in $\Jj$, hence we have $(C^B)^\Jj \subseteq (D^B)^\Jj$.
	From \cref{lemma:intFLBOTREGtointFLREG} we know that $(C^B)^\Jj = C^\Ii$ and $(D^B)^\Jj = D^\Ii$ (again, $B$ does not occur in $C$ nor in $D$), and hence $C^\Ii \subseteq D^\Ii$.
	
	Since $\Ii$ was an arbitrary model of $\Tt$, we conclude that $C \sqsubseteq_\Tt D$ holds.
	This finishes the proof of the right-to-left implication of \cref{theorem:fromFLBOTREGtoFLREG}.
	
	Next, we have lemmata used in the left-to-right implication of the proof.
	In this lemmata $\Jj$ is again an interpretation used for concept descriptions in \flreg, but now it is arbitrary, so $B$ may be interpreted as a nonempty set.
	A proof of \cref{lemma:properties-of-B} uses the following auxiliary \lcnamecref{lemma:B-word}.
	
	\begin{lemma}\label{lemma:B-word}
		For every word $w\in\SigmaR^*$ and every concept description $M$, in every model $\Jj$ of $\Rr$ in which the subsumption $B\sqsubseteq M$ holds,
		the subsumption $B\sqsubseteq\forall w.M$ holds as well.
	\end{lemma}
	
	\begin{proof}
		We prove this by induction on the length of $w$.
		In the base case of $w=\varepsilon$ the subsumption $B\sqsubseteq\forall w.M$ holds in $\Jj$ by assumption because $\forall\varepsilon.M\equiv M$.
		A longer word $w$ can be written as $w=rv$, where $r\in\SigmaR$ is a single role name, and $v$ a shorter word.
		The induction hypothesis says that $B\sqsubseteq\forall v.M$ holds in $\Jj$, hence also $\forall r.B\sqsubseteq\forall r.\forall v.M$ by monotonicity of value restrictions.
		The right side can be simplified to the desired form due to the equivalence $\forall r.\forall v.M\equiv\forall rv.M$.
		Simultaneously, in $\Rr$ we have an axiom $B \sqsubseteq \forall r.B$, which holds in $\Jj$; this allows us to conclude by transitivity of subsumption.
	\end{proof}

	\lemmaPropertiesOfB*
	
	\begin{proof}
		We proceed by induction on the structure of $K$; we have several cases depending on the shape of $K$.
		
		The thesis is clear if $K=B$, and if $K=A$ is a concept name from \SigmaC, then the subsumption holds in $\Jj$ because there is an axiom $B \sqsubseteq A$ in $\Rr$.
		
		If $K = \top$, then the subsumption holds because $\top$ is interpreted as the whole universe.
		
		If $K = K_1 \sqcap K_2$, then by the induction hypothesis both $B^\Jj \subseteq K_1^\Jj$ and $B^\Jj \subseteq K_2^\Jj$, hence also $B^\Jj\subseteq K_1^\Jj \cap K_2^\Jj = K^\Jj$.
		
		Finally, suppose that $K=\forall E.M$.
		According to \cref{eq:in-restriction}, in order to prove $B^\Jj\subseteq (\forall E.M)^\Jj$, it is enough to show that $B^\Jj\subseteq(\forall w.M)^\Jj$ for every word $w\in \Ll(E)$.
		Take such a word $w$.
		It consists of role names from \SigmaR, since only such role names occur in $E$.
		Moreover, by the induction assumption we have $B^\Jj\subseteq M^\Jj$.
		We can thus conclude using \cref{lemma:B-word}.
	\end{proof}

	We now switch to a proof of \cref{lemma:intFLREGtointFLBOTREG}, which uses the following auxiliary \lcnamecref{lemma:B4}.
	
	\begin{lemma}\label{lemma:B4}
		Let $\Jj$ be a model of $\Rr$ and $\Ii$ an interpretation such that
		\begin{compactitem}
			\item $\univ^\Ii = \univ^\Jj\setminus B^\Jj$,
			\item $A^\Ii = A^\Jj\setminus B^\Jj$ for each $A\in\names$, and
			\item $r^\Ii = \{(x,y)\in r^\Jj\mid x\notin B^\Jj\land y\notin B^\Jj\}$ for each $r\in\roles$.
		\end{compactitem}
		Then, for every word $w\in\SigmaR^*$ and every concept description $M$ involving only concept names from \SigmaC and role names from \SigmaR
		such that $M^\Ii = (M^B)^\Jj\setminus B^\Jj$, we also have $(\forall w.M)^\Ii = (\forall w.M^B)^\Jj\setminus B^\Jj$
	\end{lemma}
	
	\begin{proof}
		We proceed by induction on the length of $w$.
		In the base case of $w=\varepsilon$ the thesis follows immediately from the assumption, because $\forall\varepsilon.M\equiv M$ and $\forall\varepsilon.M^B\equiv M^B$.
		Suppose thus that $w$ is longer, and write it as $w=rv$, where $r\in\SigmaR$ is a single role name, and $v$ a shorter word.
		Denoting $K=\forall v.M$, we have $K^B=\forall v.M^B$.
		With this notation, $\forall w.M\equiv\forall r.K$ and $\forall w.M^B\equiv\forall r.K^B$.
		Moreover, the induction hypothesis says that $K^\Ii=(K^B)^\Jj\setminus B^\Jj$.
		Knowing this, we have to prove that $(\forall r.K)^\Ii=(\forall r.K^B)^\Jj\setminus B^\Jj$.
		
		Clearly $(\forall r.K)^\Ii\subseteq \univ^\Ii=\univ^\Jj\setminus B^\Jj$ and $(\forall r.K^B)^\Jj\subseteq\univ^\Jj$,
		hence in order to show $(\forall r.K)^\Ii = (\forall r.K^B)^\Jj \setminus B^\Jj$,
		it is sufficient to prove that for every element $x \in \univ^\Jj\setminus B^\Jj$
		we have the equivalence $x \in (\forall r.K)^\Ii\Longleftrightarrow x \in (\forall r.K^B)^\Jj$.
		
		Recalling the semantics of value restriction we see that $x \in (\forall r.K)^\Ii$ precisely when
		\[
		y\in K^\Ii\mbox{ for all }y\mbox{ such that }(x,y)\in r^\Ii.
		\]
		By assumptions of the lemma, $(x,y)\in r^\Ii$ precisely when $(x,y)\in r^\Jj$ and $x\notin B^\Jj$ and $y\notin B^\Jj$.
		Since only $x \in \univ^\Jj\setminus B^\Jj$ are considered, the condition $x\notin B^\Jj$ can be dropped here.
		Going further, the induction hypothesis $K^\Ii=(K^B)^\Jj\setminus B^\Jj$ implies that for $y\notin B^\Jj$
		the condition $y \in K^\Ii$ is equivalent to $y \in (K^B)^\Jj$.
		Thus we obtain that $x \in (\forall r.K)^\Ii$ precisely when
		\[
		y \in (K^B)^\Jj\mbox{ for all }y\mbox{ such that }(x,y) \in r^\Jj\mbox{ and }y\not\in B^\Jj.\tag{$\star$}
		\]
		On the other hand, $x \in (\forall r.(K^B))^\Jj$ precisely when
		\[
		y \in (K^B)^\Jj\mbox{ for all }y\mbox{ such that }(x,y) \in r^\Jj.\tag{$\star\star$}
		\]
		
		It remains to show that Conditions ($\star$) and ($\star\star$) are equivalent, for every $x \in \univ^\Jj\setminus B^\Jj$.
		Every $y$ considered in Condition ($\star$) is also considered in Condition ($\star\star$), so Condition ($\star\star$) implies Condition~($\star$).
		
		For a proof of the opposite implication assume that Condition ($\star$) holds and take an arbitrary $y$ such that $(x,y) \in r^\Jj$.
		In order to obtain Condition ($\star\star$), we must show that $y \in (K^B)^\Jj$.
		If $y \notin B^\Jj$, then Condition ($\star$) directly yields $y \in (K^B)^\Jj$.
		On the other hand, by \cref{lemma:properties-of-B} we know that $B^\Jj \subseteq (K^B)^\Jj$, so if $y \in B^\Jj$, then $y \in (K^B)^\Jj$ as well.
		This finishes the proof of the equivalence, hence of the whole lemma.
	\end{proof}

\lemmaintFLREGtointFLBOTREG*

\begin{proof}
	We proceed by induction on the structure of $K$; we have several cases depending on the shape of $K$.
	
	Let $K = A$ be a concept name.
	By assumption we have $A^\Ii = A^\Jj\setminus B^\Jj$.
	Since $K^B = A$, it follows that 
	\[
	K^\Ii = A^\Ii = A^\Jj\setminus B^\Jj = (K^B)^\Jj\setminus B^\Jj.
	\]
	
	Let $K = \top$.
	Then we have $\univ^\Ii=\univ^\Jj\setminus B^\Ii$ by assumption, as well as $K^B = \top$,
	hence
	\[
	K^\Ii = \top^\Ii = \univ^\Ii = \univ^\Jj\setminus B^\Jj = \top^\Jj\setminus B^\Jj = (K^B)^\Jj\setminus B^\Jj.
	\]
	
	Let $K = \bot$.
	Then we have $K^B = B$, hence
	\[
	K^\Ii = \bot^\Ii = \emptyset = B^\Jj\setminus B^\Jj = (K^B)^\Jj\setminus B^\Jj.
	\]
	
	Let $K = K_1 \sqcap K_2$.
	By the induction hypothesis we know that $K_1^\Ii = (K_1^B)^\Jj\setminus B^\Jj$ and $K_2^\Ii = (K_2^B)^\Jj\setminus B^\Jj$.
	Since $K^B=K_1^B\sqcap K_2^B$, we obtain
	\begin{samepage}\begin{align*}
			K^\Ii &= (K_1 \sqcap K_2)^\Ii = K_1^\Ii \cap K_2^\Ii = \bigl((K_1^B)^\Jj\setminus B^\Jj\bigr)\cap\bigl((K_2^B)^\Jj\setminus B^\Jj\bigr) \\
			& =(K_1^B)^\Jj \cap (K_2^B)^\Jj\setminus B^\Jj =(K_1^B\sqcap K_2^B)^\Jj\setminus B^\Jj=(K^B)^\Jj\setminus B^\Jj.
		\end{align*}
	\end{samepage}
	
	Finally, let $K = \forall E.M$.
	Then $K^B=\forall E.M^B$.
	From \cref{eq:in-restriction} we know that $x\in(\forall E.M)^\Ii$ precisely when $x\in(\forall w.M)^\Ii$ for all words $w\in\Ll(E)$.
	Likewise, $x\in(\forall E.M^B)^\Jj\setminus B^\Jj$ precisely when $x\in(\forall w.M^B)^\Jj\setminus B^\Jj$ for all words $w\in\Ll(E)$.
	It is thus enough to prove for every word $w\in\Ll(E)$ that $(\forall w.M)^\Ii=(\forall w.M^B)^\Jj\setminus B^\Jj$.
	Note that words from $\Ll(E)$ consist of role names from \SigmaR, since only such role names occur in $E$.
	Moreover, the induction hypothesis implies $M^\Ii = (M^B)^\Jj\setminus B^\Jj$.
	We thus obtain the desired equality $(\forall w.M)^\Ii=(\forall w.M^B)^\Jj\setminus B^\Jj$ from \cref{lemma:B4}.
	\qed\end{proof}

We can now finish the proof of \cref{theorem:fromFLBOTREGtoFLREG}, showing its left-to-right implication:
we assume that $C \sqsubseteq_\Tt D$ holds, and we prove that then $C^B \sqsubseteq_{\Tt^B \cup \Rr} D^B$ holds. To this end, take an arbitrary model $\Jj$ of $\Tt^B \cup \Rr$.

First consider the case where $B^\Jj=\univ^\Jj$. Since $D^B$ is a concept description involving only concept names from $\SigmaC\cup\{B\}$ and role names from $\SigmaR$, \cref{lemma:properties-of-B} yields $B^\Jj \subseteq (D^B)^\Jj$. Because $B^\Jj=\univ^\Jj$, it follows that $(D^B)^\Jj=\univ^\Jj$. Hence $(C^B)^\Jj \subseteq (D^B)^\Jj$ holds trivially.

It remains to consider the case where $B^\Jj\neq \univ^\Jj$. Then $\univ^\Jj\setminus B^\Jj\neq\emptyset$, and thus the interpretation $\Ii$ from \cref{lemma:intFLREGtointFLBOTREG} is well-defined. Consider a corresponding interpretation $\Ii$ as specified in \cref{lemma:intFLREGtointFLBOTREG}.

In order to see that $\Ii$ is a model of $\Tt$, consider any axiom $(K\sqsubseteq M)\in\Tt$. We have $(K^B \sqsubseteq M^B) \in \Tt^B$, and since $\Jj$ is a model of $\Tt^B$, it follows that $(K^B)^\Jj \subseteq (M^B)^\Jj$. By \cref{lemma:intFLREGtointFLBOTREG}, we obtain $K^\Ii=(K^B)^\Jj\setminus B^\Jj\text{ and } M^\Ii=(M^B)^\Jj\setminus B^\Jj,$ where $K$ and $M$ involve only concept names from $\SigmaC$ and role names from $\SigmaR$. Thus $K^\Ii \subseteq M^\Ii,$ and so the axiom $(K\sqsubseteq M)$ holds in $\Ii$. Since the axiom was arbitrary, we conclude that $\Ii$ is a model of $\Tt$.

By assumption, the subsumption $C\sqsubseteq_\Tt D$ holds, and therefore, since $\Ii$ is a model of $\Tt$, we have $C^\Ii \subseteq D^\Ii$.

Applying \cref{lemma:intFLREGtointFLBOTREG} to $C$ and $D$, we obtain $(C^B)^\Jj\setminus B^\Jj = C^\Ii \subseteq D^\Ii = (D^B)^\Jj\setminus B^\Jj$.

On the other hand, \cref{lemma:properties-of-B} implies that $B^\Jj \subseteq (D^B)^\Jj$.

We now show that $(C^B)^\Jj \subseteq (D^B)^\Jj$. Take any $x\in (C^B)^\Jj$. If $x\in B^\Jj$, then $x\in (D^B)^\Jj$ by the above inclusion. If $x\notin B^\Jj$, then $x\in (C^B)^\Jj\setminus B^\Jj \subseteq (D^B)^\Jj\setminus B^\Jj,$ and hence again $x\in (D^B)^\Jj$. Thus every element of $(C^B)^\Jj$ belongs to $(D^B)^\Jj$, that is, $(C^B)^\Jj \subseteq (D^B)^\Jj$.

Since $\Jj$ was chosen to be an arbitrary model of $\Tt^B \cup \Rr$, the subsumption $C^B \sqsubseteq_{\Tt^B \cup \Rr} D^B$ holds, as required.


\end{document}